\documentclass[a4paper,UKenglish,cleveref, autoref, thm-restate]{lipics-v2021}
\nolinenumbers

\usepackage{booktabs} 
\usepackage{todonotes}
\usepackage{algorithm}
\usepackage[noend]{algpseudocode}
\usetikzlibrary{arrows.meta, positioning, calc, shapes}

\title{A Separator-based Algorithm for the Graph Edit Distance Problem}

\author{Laura B{\"u}lte\footnote{corresponding author}}{University of Bonn, Germany}{lbuelte@uni-bonn.de}{https://orcid.org/0009-0002-2697-4076}{} 

\author{Philip Mayer}{University of Bonn, Germany}{pmayer@uni-bonn.de}{https://orcid.org/0009-0007-4800-7753}{}

\author{Lars M{\"u}ller\texorpdfstring{\footnotemark[1]}{}}{University of Bonn, Germany}{lars-mueller@uni-bonn.de}{https://orcid.org/0009-0003-7656-227X}{}

\author{Petra Mutzel}{University of Bonn, Germany}{pmutzel@uni-bonn.de}{https://orcid.org/0000-0001-7621-971X}{}

\authorrunning{L. B{\"u}lte, P. Mayer, L. M{\"u}ller, P. Mutzel} %

\Copyright{Laura B{\"u}lte, Philip Mayer, Lars M{\"u}ller, Petra Mutzel} 

\ccsdesc[500]{Theory of computation~Parameterized complexity and exact algorithms}
\ccsdesc[500]{Mathematics of computing~Graph algorithms}

\keywords{Graph Edit Distance, Graph Similarity, Exact Algorithms, Recursive Algorithms, Graph Separators, Exponential Time Algorithms}

\funding{This research was partially funded by the Deutsche Forschungsgemeinschaft (DFG, German Research Foundation) under grant FOR-5361 -- 459420781.}

\acknowledgements{We thank Dr.\ Tamás Horváth for his insightful comments.}

\DeclareMathOperator{\NEC}{c_V}
\newcommand{\ANEC}[1]{\mathop{\mathrm{c_V^{#1}}}} 
\DeclareMathOperator{\EEC}{c_E}
\DeclareMathOperator{\IEC}{c}
\DeclareMathOperator{\GED}{\mathrm{GED}}
\DeclareMathOperator{\dom}{\mathrm{dom}} 
\DeclareMathOperator{\im}{\mathrm{im}} 
\newcommand{\posreals}{\mathbb{R}_{\geq 0}}
\newcommand{\PE}[1]{\binom{#1}{2}} 
\newcommand{\floor}[1]{\left \lfloor #1 \right \rfloor}
\newcommand{\ceil}[1]{\left \lceil #1 \right \rceil}
\newcommand{\restr}[2]{\left.#1\right|_{#2}}
\newcommand{\bij}[2]{\mathrm{Bij}(#1, #2)}
\newcommand{\inj}[2]{\mathrm{Inj}(#1, #2)}
\newcommand{\bigO}[1]{\mathcal{O}\!\left(#1\right)}
\newcommand{\notbigO}[1]{o\!\left(#1 \right)}
\newcommand{\expObase}{\mathcal{O}^*}
\newcommand{\expO}[1]{\expObase\!\left(#1\right)}
\newcommand{\classP}{\mathsf{P}}
\newcommand{\classNP}{\mathsf{NP}}
\newcommand{\AStar}{A$^*$}

\DeclareMathOperator{\deltaEC}{c_{\delta}}

\newcommand{\AEC}[1]{\mathop{\mathrm{c}^{#1}}}
\newcommand{\AGED}[1]{\GED^{#1}}
\DeclareMathOperator{\deltaECplus}{c_{\delta}^+}
\DeclareMathOperator{\impsi}{\im(\psi) }
\DeclareMathOperator{\impsione}{\im(\psi_1)}
\DeclareMathOperator{\impsitwo}{\im(\psi_2)}
\DeclareMathOperator{\dompsi}{\dom(\psi)}
\newcommand{\mindown}[2]{\min_{#1} \left( #2 \right)}
\newcommand{\minback}[2]{\min \left\{ #2 : #1 \right\}}
\newcommand{\psicomplement}{\chi}
\newcommand{\set}[1]{\left\{ #1 \right\}}
\newcommand{\alphasbalsep}[2]{\ensuremath{(#1,#2)}-balanced-separable}

\begin{document}

\maketitle

\begin{abstract}
The Graph Edit Distance (GED) is a widely used graph similarity measure asking for the minimum cost of a sequence of edits transforming one (labeled) graph into another. The considered edit operations are deletion, insertion, and substitution of nodes and edges.
Special cases include the Graph Isomorphism problem, as well as many other graph problems that ask for the existence or minimum cost of a certain substructure, like the Traveling Salesman or Maximum Clique problem.

We present a novel exponential time algorithm to compute the exact GED and a corresponding edit sequence in $\expO{(4 + \varepsilon)^n}$ time and polynomial space, provided one of the two graphs admits strictly sublinear balanced separators. 
In particular, the claimed runtime holds if one of the graphs is $K_h$-minor-free (e.g., planar), or has bounded treewidth, which is the case for many real-world applications (e.g., all instances in GEDLIB).
This substantially improves the best known worst-case running time bounds of $\expO{n!}$ for these graph classes.
\end{abstract}

\newpage

\section{Introduction}
\label{sec:typesetting-summary}

The Graph Edit Distance (GED) is a flexible and widely used graph similarity measure. It poses the simple question: How much does it cost to transform one graph into another?
Due to its flexibility, the GED has found applications across a broad range of domains, including cancer detection and the alignment of protein–protein interaction networks \cite{appl-gedevo}, biometric identification \cite{appl-biom-hand,appl-biom-retina}, molecular similarity scoring~\cite{appl-mol-sim}, rational drug design \cite{appl-drug}, as well as malware detection and classification \cite{appl-bin-cmp,appl-adg}.

Classically, the GED problem is defined as follows \cite{ged-first-intro}: We are given two node- and edge-labeled graphs, and the costs of the (elementary) edit operations -- inserting and deleting isolated nodes, inserting and deleting edges, and substituting nodes and edges. The task is to find the minimum cost of a sequence of edit operations transforming one graph into the other. 
For algorithmic purposes, this definition is not well-suited. 
Hence, most approaches use a definition based on node mappings \cite{bunke-etgraphmatch}.
 Given graphs $G$ and $H$, we ask for the minimum cost of an injective mapping of (a subset of) the nodes of $G$ to (a subset of) the nodes of $H$. A mapping incurs the following costs:
(1) mapped nodes incur substitution costs;
unmapped nodes in $G$ and $H$ incur deletion and insertion costs, respectively;
(2) for a pair of edges $e\in E(G)$, $f \in E(H)$, if the endpoints of $e$ are mapped to the endpoints of $f$, we say $e$ and $f$ are mapped, and they incur substitution costs;
if $e$ is not mapped to any edge in $H$, it incurs deletion costs,
and if $f$ is not mapped to any edge in $G$, it incurs insertion costs.

The two formulations have a natural correspondence between their cost functions.
While every node map can be transformed into a sequence of edit operations, the converse does not generally hold.
However, if the elementary edit costs satisfy the properties of a metric, any optimum sequence of edits corresponds to an optimum node map \cite{ilp-mip-survey}.
Consequently, the GED problem formulation on edit sequences with metric costs is equivalent to the GED problem formulation on node maps. 
In the next section, we present a more formal definition of the GED based on node mappings, which is the formulation used in our algorithm.

Exact approaches for GED computation usually resort to informed tree searches (e.g.,\ \cite{ts-astar-old,ts-astar-bmao}) or solving ILPs (e.g.,\ \cite{ilp-justice,ilp-fori}). 
Although these approaches may perform well on real-world data, they typically lack strong, provable worst-case running time guarantees.
Naturally, the question arises for which cases we can find exponential time algorithms that beat a complete exploration of the search space in $\expO{n!}$ worst-case running time. In this work we give a positive answer for several graph classes.

\subparagraph*{Our Contribution.}

We propose SR-GED, a novel recursive algorithm for the exact computation of the GED that exploits small balanced separators in one input graph.
To the best of our knowledge, SR-GED achieves a new best worst-case running time of $\expO{(4+\varepsilon)^n}$ (for any fixed $\varepsilon >0$) on many graph classes, such as planar graphs, $K_h$-minor-free graphs, or graphs with bounded treewidth, using only polynomial space.
The recursive nature of SR-GED makes it suitable for combination with practically efficient alternative approaches. Furthermore, it can easily be adapted to provide a node map minimizing the edit costs.

\section{Related Work}

\subparagraph*{Theoretical Results.}

Computing the GED is an NP-hard problem, even for unit edit costs, by a reduction from the Subgraph Isomorphism (SGI) problem \cite{apx-stars}. Computing the GED remains NP-hard even if one of the two graphs is a cycle or path, by reductions from Hamiltonian Cycle and Hamiltonian Path, respectively, to SGI.
Furthermore, unless the exponential time hypothesis (ETH) is false, SGI cannot be decided in time $n^{\notbigO{n}}$ \cite{sgi-hardness}, thus neither can the GED be computed in time $n^{\notbigO{n}}$ in general. This is essentially the asymptotic runtime of a brute-force approach. Therefore, achieving better runtimes necessitates a restriction of the problem.
Yet, even if \emph{both} graphs are restricted to general trees \cite{hardness-trees}, the GED remains NP-hard. Hence, no FPT algorithm parameterized by treewidth exists unless $\classP = \classNP$. However, an exponential time algorithm with running time in $\bigO{2^{2n/3}}$ is known for the general tree-case \cite{exptime-trees}.
A more restricted problem on ordered trees, the Tree Edit Distance can be solved in polynomial time \cite{ted-survey}.
The GED is also hard to approximate, even for unit edit costs: No PTAS can exist unless $\classP = \classNP$ \cite{lin-apx-hardness}, and there is no polynomial-time approximation algorithm with any finite ratio unless Graph Isomorphism is in $\classP$ \cite{blumenthal-thesis}.

Moreover, several relevant problems can be reduced to GED. The Maximum Common Subgraph problem on labeled graphs is equivalent to a GED problem where insertions and deletions have cost $1$ and substitutions cost $2$ \cite{ged-mcs}. A symmetric Quadratic Assignment problem (QAP) formulation (Koopmans and Beckmann \cite{kb-qap}) can be reduced to a GED instance of the same size in polynomial time (see Appendix \ref{appendix:qap}). 
Finally, the TSP on a graph~$G$ with edge weights $w(e)$ can be reduced to a GED instance by mapping a cycle of size $n$ to~$G$: Deleting any edge of the cycle is made prohibitively expensive, and substituting an edge on the cycle with an edge $e \in E(G)$ is given cost $w(e)$. Other costs are set to zero. For TSP, the worst-case time complexity of $\bigO{2^n \cdot n^2}$ due to Held and Karp \cite{held-karp} has resisted improvement for decades.
Thus, we are unlikely to find an algorithm computing the GED in time $\expO{(2-\delta)^n}$ for some $\delta > 0$ even if one graph is a cycle.

\subparagraph*{Practical Approaches.}

Algorithms for exact GED computation can mainly be divided into two categories: informed tree searches and integer linear programming (ILP) formulations.

Tree searches enumerate node mappings between the graphs in a search tree, starting from the empty mapping (an exception is CSI-GED, which maps edges instead \cite{ts-csi}). 
Each level of the tree extends the current mapping by assigning one additional node. Then, informed search strategies use lower and upper bounds to guide the exploration.
Early methods used bipartite matching heuristics \cite{ts-astar-old} to compute lower bounds, and numerous improvements have been proposed since.
Approaches use Best-First Search (\AStar) \cite{ts-astar-old,ts-iso-nodes,ts-astar-lsa,ts-astar-bmao}, Depth-First Search \cite{ts-dfs}, or Beam Stack Search \cite{ts-bss}. 
Additionally, node-isomorphisms have been leveraged to reduce the search space \cite{ts-iso-nodes,ts-bss}. Significant progress has been made in improving both the tightness and efficiency of lower bounds \cite{ts-bss,ts-astar-lsa,ts-astar-bmao}. State-of-the-art methods such as \AStar-BMao \cite{ts-astar-bmao} can solve benchmark instances with around 50 nodes in reasonable time.

Complementary to tree search methods, integer programming approaches encode node mappings as binary variables. From early formulations \cite{ilp-justice,ilp,ilp-mip-survey} to recent models \cite{ilp-fori}, their scalability has improved substantially, with the latter solving instances with up to 100 nodes to optimality and outperforming state-of-the-art tree search techniques such as \AStar-BMao~\cite{ts-astar-bmao} by several orders of magnitude \cite{ilp-fori}.

Since exact computation is still intractable for large instances, practical applications typically estimate the GED by solving minimum weight bipartite matching problems \cite{apx-stars,ged-applications}.

\section{Preliminaries} \label{prelims}

Let $G = (V, E)$ and $H = (W, F)$ be two simple, undirected graphs.
For clarity of presentation, we only consider pairs of simple graphs where $n := |V| = |W|$. 
We take the vertices of both graphs to be enumerated as $V = W = \{1, \ldots, n\}$,
which gives $E,F\subseteq\PE{V}=\PE{W}$.

We start this section with a formal definition of the Graph Edit Distance (GED) problem based on node maps. 
We require the mapping costs to be \emph{reasonable} in the following way: For a pair of nodes $v\in V$ and $w\in W$, the cost of substituting $v$ by $w$ must not exceed the cost of deleting $v$ and inserting $w$ combined. Analogously, for a pair of edges $e \in E$ and $f \in F$, the cost of substituting $e$ by $f$ must not exceed the cost of deleting $e$ and inserting $f$ combined.
Then, since $G$ and $H$ have the same number of nodes, there is an optimum node map that substitutes all nodes. It therefore suffices to consider optimum solutions as bijective mappings between $V$ and $W$.

\begin{definition}[Node Map]
A node map is a bijection $\phi: V \to W$ between the nodes of the two graphs.
We denote $\phi[X] := \{\phi(x) : x \in X\}$ for $X \subseteq V$.
A bijection $\psi: V' \to W'$ between two subsets $V' \subseteq V$ and $W' \subseteq W$ is called a partial node map.
\end{definition}

\begin{definition}[Edit Costs]
A node edit cost function $\NEC: V \times W \to \posreals$ denotes the cost incurred when editing a node $v$ in $G$ to match a node $w$ in $H$.
An edge edit cost function $\EEC: \PE{V} \times \PE{W} \to \posreals$ denotes the cost incurred when editing an edge in $G$ to match an edge in $H$.
\end{definition}

Since $V = W$ and $\PE{V} = \PE{W}$, $\NEC$ and $\EEC$ can be thought of as square matrices with dimensions $n$ resp.\ $\binom{n}{2}$.
We explicitly allow all pairs in $\PE{V}$ and $\PE{W}$ in the edge edit costs. Hence, the edge edit costs incorporate substitutions, insertions and deletions. More specifically, for potential edges $e\in \PE{V} $ and $f\in \PE{W} $ we get the following edge edit costs:

\begin{itemize}
    \item If $e \in E$ and $f \in F$, we take $\EEC(e, f)$ to be the cost of substituting $e \to f$.
    \item If $e \in E$ but $f \notin F$, we take $\EEC(e, f)$ to be the cost of deleting $e$, also denoted $\EEC(e, \emptyset)$.
    \item If $e \notin E$ but $f \in F$, we take $\EEC(e, f)$ to be the cost of inserting $f$, also denoted $\EEC(\emptyset, f)$.
    \item If $e \notin E$ and $f \notin F$, we define $\EEC(e, f) := 0$ for ease of notation: 
    If neither edge is present, no edit is necessary.
\end{itemize}

Note that we do not define costs for node deletion or insertion, since they will never be incurred by an optimum solution anyway.

Based on these definitions, we define the edit costs induced by a node map as follows.

\begin{definition}[Induced Edit Costs]\label{def:IEC}
Let $\phi\in\bij{V}{W}$ be a node map. Let $\NEC(\phi) := \sum_{v \in V} \NEC(v, \phi(v))$ and $\EEC(\phi) := \sum_{e \in \PE{V}} \EEC(e, \phi[e])$ be the induced node and edge edit costs respectively.
Then the induced edit costs of $\phi$ are defined as
\begin{equation*}
    \IEC(\phi) := \NEC(\phi) + \EEC(\phi) .
\end{equation*}
\end{definition}

Definition \ref{def:IEC} also applies to partial node maps $\psi\in\bij{V'}{W'}$ for $V'\subseteq V, W'\subseteq W$.
Intuitively, $\IEC(\phi)$ is the cost of edits which, when applied to $G$, make $\phi$ an isomorphism between $G$ and $H$.
When computing the GED, we strive to find a node map that is already as close to an isomorphism as possible.

\begin{definition}[Graph Edit Distance]
The Graph Edit Distance (GED) between $G$ and $H$ is defined as the minimum induced edit cost over all possible node maps:
\begin{equation*}
    \GED(G, H, \NEC, \EEC) := \min \{\IEC(\phi) : \phi \in \bij{V}{W}\} 
\end{equation*}
\end{definition}

For some algorithms, for example, \AStar-based graph searches and also our own algorithm, it is sometimes useful to define the GED if a partial node map is already fixed.

\begin{definition}[Restricted Graph Edit Distance] \label{def:restrictedGED}
    Let $U\subseteq V$. For a fixed partial node map $\psi\in\inj{U}{W}$, we define the restricted GED by
    \begin{equation*}
       \AGED{\psi}(G,H,\NEC, \EEC) := \minback{\phi\in\bij{V}{W}, \restr{\phi}{U}=\psi}{\IEC(\phi)} .
    \end{equation*}
\end{definition}

We will usually suppress $\NEC$ and $\EEC$, writing just $\GED(G, H)$.
However, in the course of our algorithm we adapt the node edit costs $\NEC$, thus we need to be explicit about them here.
In addition to finding the GED, it is often useful to also determine a mapping $\phi$ that minimizes $\IEC(\phi)$. 
The algorithm presented here straightforwardly permits this. 

In Appendix~\ref{appendix:restrictions}, we show that all the restrictions we imposed in our problem formulation can be removed with a polynomial-time preprocessing. 

Next, we define the notion of a separator, which is at the core of our algorithm.

\begin{definition}[Balanced Separator]
Let $G$ be a graph and $S \subseteq V(G)$. The set $S$ is called an $\alpha$-separator of $G$ if the graph $G - S$ is the disjoint union of two graphs $G_1$ and $G_2$ such that
\(
\max\{\,|V(G_1)|,\ |V(G_2)|\,\} \le \alpha |V(G)|.
\)
The parameter $\alpha$ is called the balance factor.
\end{definition}

Another definition of a balanced separator allows the decomposition to consist of multiple connected components $G_1, \dots, G_k$ such that $\max_i |V(G_i)| \leq \alpha n$, but in our work we explicitly require the decomposition to consist of exactly two (possibly non-connected) components. 

We will later (cf.\ \Cref{thm:runtime}) consider graphs with a bound on the size of balanced separators which also extends to subgraphs. This is captured in the following definition:

\begin{definition}[\alphasbalsep{\alpha}{s}]
Let $s: \mathbb{N} \to \mathbb{R}_{\ge 0}$ be a  function and let $\alpha \in (0,1)$.
\begin{itemize}
    \item A graph $G$ is \alphasbalsep{\alpha}{s} if for every node-induced subgraph $G[X], X \subseteq V(G)$, there exists an $\alpha$-separator $S \subseteq X$ of $G[X]$ such that \(|S| \le s(|X|)\).
\item A class $\mathcal{G}$ of graphs is called \alphasbalsep{\alpha}{s} if every graph $G \in \mathcal{G}$ is \alphasbalsep{\alpha}{s}.
\end{itemize}
\end{definition}

Finally, for the runtime analysis, we need the following notions:
Let $f: \mathbb{N} \to \mathbb{R}$ be a function. We call (the order of growth of) $f$ \emph{strictly sublinear} if $f = \bigO{n^{1 - \delta}}$ for some fixed $\delta \in (0, 1]$. 
Further, we write $f(n) = \expO{g(n)}$ if $f(n) = \bigO{\mathrm{poly}(n) \cdot g(n)}$ to suppress polynomial factors in the runtime of exponential time algorithms.

\section{Separator-based Graph Edit Distance Computation}

Our algorithm is inspired by classical separator-based divide-and-conquer schemes. In these approaches, one finds balanced separators and recursively splits the instance into smaller subinstances until they are small enough to be solved by enumeration. 
For the GED, the problem is more difficult because instead of a single graph $G$, we need to also consider a second graph $H$, and more importantly, the set of possible node maps $\phi$ relating the two. 

At each recursive step, we can still divide $G$ as usual. However, we must ensure that the separator $S$ is mapped optimally. 
This requires enumerating all partial node maps $\psi$ of $S$ to $V(H)$. For each of these, we then need to find an optimum mapping $\psi_r$ for the remaining  graph $G_r=G-S$, i.e., we must also enumerate all possible mappings $\psi_r$ interlocked with~$\psi$.\begin{figure}[tbp]
    \centering
    \begin{minipage}[t]{0.46\textwidth}
        \centering
        \includegraphics[width=0.9\linewidth]{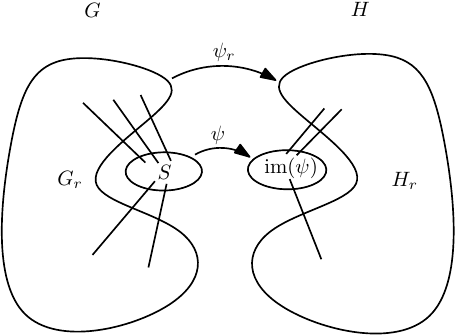}
        \captionsetup{textformat=simple}\caption{Given $S$ and a fixed $\psi$, the edge edit cost of $\delta_G(S)$ and $\delta_H(\text{im}(\psi))$  can be absorbed into the node edit costs for $G_r$ and $H_r$.}\label{fig:adapted_costs}
    \end{minipage}
    \hfill
    \begin{minipage}[t]{0.46\textwidth}
        \centering
        \includegraphics[width=0.9\linewidth]{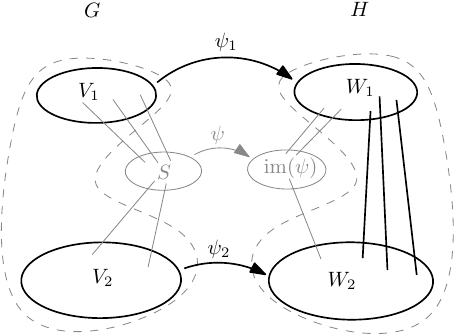}
        \caption{If we remove a separator $S$, $G[V_1]$ and $G[V_2]$ are disjoint, but $H[W_1]$ and $H[W_2]$ are not necessarily disjoint. Still $\psi_1$ and $\psi_2$ are independent if $W_1$ and $W_2$ are fixed.}\label{fig:disjoint_graphs}
    \end{minipage}
\end{figure}

The main results of this section show that the latter enumeration can be carried out more efficiently. First, for a fixed $\psi$, we adapt the node edit costs to account for the edges in the cut of $G_r$ in $G$ and the cut of the corresponding $H_r$ in $H$ (see Figure~\ref{fig:adapted_costs}). Following that, the optimum $\psi_r$ can be determined independently of $\psi$. Second, when $G_r$ decomposes into two disjoint subgraphs $G_1$ and $G_2$ (as occurs when a separator is removed), only the images in $H_r$ of the partial node maps $\psi_1$ and $\psi_2$ must be compatible (see Figure~\ref{fig:disjoint_graphs}). Consequently, once the images $W_1$ and $W_2$ of the mappings are fixed in $H$, the optimum mappings from $G_1$ to $H[W_1]$ and from $G_2$ to $H[W_2]$ can be computed independently, reducing the size of the subinstances significantly. If the size of the subproblems drops geometrically due to balanced separators, and the separators are small enough such that the cost of enumerating the partial mappings $\psi$ is tolerable, we obtain a runtime of $2^{\bigO{n}}$.

We begin this section by defining the \emph{adapted node edit costs} and discussing their implications for the restricted GED. Further, we explain how instances consisting of disjoint graphs can be handled more efficiently. We then give a detailed description of our algorithm. Finally, we combine these results to show correctness, and provide a detailed runtime analysis.

\subsection{Adapted Node Edit Costs}

The algorithm adapts node edit costs to enable a reduction to independent subproblems despite the separator connecting the subgraphs of~$G$.
For the remainder of this section, we consider simple, undirected graphs $G = (V, E)$ and $H = (W, F)$ with $n := |V| = |W|$.

\begin{definition}[Cut Edit Costs] \label{def:deltaec}
    Let $\phi\in\bij{V}{W}$, $V_1\sqcup V_2=V$, and $\psi_i := \restr{\phi}{V_i}$ for $i=1,2$. 
    We define the cut edit costs as
    \begin{equation*}
        \deltaEC(\psi_1, \psi_2) := \sum_{v_1\in V_1}\sum_{v_2\in V_2}\EEC(\{v_1,v_2\},\{\psi_1(v_1),\psi_2(v_2)\}) .
    \end{equation*}
\end{definition}

Let $W_i := \im(\psi_i)$ for $i=1,2$.
As $\phi$ is a bijection and, therefore, $W_1\sqcup W_2=W$ and $\psi_i\in\bij{V_i}{W_i}$, $\deltaEC$ accounts for the mapping costs between the cuts $\delta_G(V_1) = \delta_G(V_2)$ and $\delta_H(W_1) = \delta_H(W_2)$, where $\delta_G(U) = \set{\set{u,v}\in E(G) : u\in U, v\in V(G)\setminus U }$.
By definition of $\EEC$ this also accounts for all insertion and deletion costs for such edges.
The following lemma shows that the induced edit costs of $\phi$ split exactly into the contributions from the partial node maps $\psi_1$ and $\psi_2$, together with the cut edit costs $\deltaEC$.

\begin{lemma} \label{lemma:delatec}
    Let $\phi\in\bij{V}{W}$, $V_1\sqcup V_2=V$, and $\psi_i := \restr{\phi}{V_i}$ for $i=1,2$. We then have
\begin{equation*}
    \IEC(\phi) = \IEC(\psi_1) + \IEC(\psi_2) + \deltaEC(\psi_1, \psi_2) .
\end{equation*}
\end{lemma}

\begin{proof}
    It is clear that $\NEC(\phi)=\NEC(\psi_1) + \NEC(\psi_2)$. We further have
    \begin{align*}
        \EEC(\phi) &= \sum_{e \in \PE{V}} \EEC(e, \phi[e]) \\
        &=\sum_{e \in \PE{V_1}} \EEC(e, \psi_1[e])+\sum_{e \in \PE{V_2}} \EEC(e, \psi_2[e]) + \sum_{v_1 \in V_1, v_2 \in V_2} \EEC(\{v_1, v_2\}, \{\psi_1(v_1), \psi_2(v_2)\}) \\
        &= \EEC(\psi_1) + \EEC(\psi_2) + \deltaEC(\psi_1,\psi_2)
        \qedhere
    \end{align*}
\end{proof}

We will later attribute the costs of editing the cut edges to the node edit costs of one part of the partition. 

\begin{definition}[Adapted Node Edit Costs] \label{def:anec}
     Let $\phi\in\bij{V}{W}$, $V_1\sqcup V_2=V$, and $\psi_i := \restr{\phi}{V_i}$ for $i=1,2$. 
     We define the adapted node edit costs for $v\in V_2, w\in \impsitwo$ by
    \begin{equation*}
        \ANEC{\psi_1}(v,w) := \NEC(v,w) + \sum_{v' \in V_1}\EEC(\{v',v\},\{\psi_1(v'), w\}) ,
    \end{equation*}
    and, consistent with the definition of induced edit costs, we denote
    \begin{equation*}
        \AEC{\psi_1}(\psi_2) := \sum_{v\in V_2} \ANEC{\psi_1}(v, \psi_2(v)) + \EEC(\psi_2) .
    \end{equation*}
\end{definition}

By inserting the definition of $\AEC{\psi_1}(\psi_2)$ into Lemma~\ref{lemma:delatec}, we can split the edit costs of~$\phi$ into the sum of two partial node map costs, with the cut edit costs absorbed into one of the maps.

\begin{observation} \label{observation:aec}
    We have $\IEC(\phi) = \IEC(\psi_1) + \AEC{\psi_1}(\psi_2)$ and $\IEC(\phi) = \IEC(\psi_2) + \AEC{\psi_2}(\psi_1)$. 
\end{observation}

Using this observation, the restricted GED from Definition \ref{def:restrictedGED} can be reformulated so that $\psi$ no longer appears as part of the optimization.

\begin{lemma}\label{lemma:restricted_ged_to_induced_costs}
    Let $U \subseteq V$ and $\psi\in\inj{U}{W}$. Then: 
    \begin{equation*}
        \AGED{\psi}(G, H) = \IEC(\psi) + \GED(G - U, H - \impsi, \ANEC{\psi}, \EEC)
    \end{equation*}
\end{lemma}

\begin{proof}
    We have:
    \begin{align*}
        \AGED{\psi}(G,H) 
        = &\minback{\phi\in\bij{V}{W}, \restr{\phi}{U}=\psi}{\IEC(\phi)} \\
        = &\minback{\phi\in\bij{V}{W}, \restr{\phi}{U} = \psi}{\IEC(\psi) + \AEC{\psi}(\psicomplement)}\text{ where }\psicomplement := \restr{\phi}{V\setminus U} \\
        = &\IEC(\psi) + \minback{\psicomplement\in\bij{V\setminus U}{W\setminus \impsi}}{ \AEC{\psi}(\psicomplement)}\\
        = &\IEC(\psi) + \GED(G - U, H - \impsi, \ANEC{\psi}, \EEC)
        \qedhere
    \end{align*}
\end{proof}

For a fixed $U \subseteq V$, the GED can be reformulated by enumerating all embeddings of $U$ and computing the corresponding restricted GED.

\begin{observation}
     For arbitrary $U\subseteq V$, we have
    \[
        \GED(G,H) = \minback{\psi \in \inj{U}{W}}{\AGED{\psi}(G,H)}.
    \]
\end{observation}

Applying Lemma~\ref{lemma:restricted_ged_to_induced_costs} then yields:
\begin{corollary} \label{lemma:ged-embed}
    Let $U \subseteq V$ be arbitrary. We have
    \begin{equation*}
        \GED(G,H) = \minback{\psi \in \inj{U}{W}}{\IEC(\psi) + \GED(G - U, H - \impsi, \ANEC{\psi}, \EEC)} .
    \end{equation*}
\end{corollary}
Two further properties allow Corollary~\ref{lemma:ged-embed} to be applied effectively:
\begin{enumerate}[(i)]
    \item The set $U$ is small enough that it is feasible to enumerate all $\psi \in \inj{U}{W}$.
    \item Computing $\GED(G - U, H - \impsi, \ANEC{\psi}, \EEC)$ on the remaining graph becomes easier.
\end{enumerate}
We later show that both of these properties are satisfied if we choose the set $U$ to be a sufficiently small balanced separator in $G$.

\begin{remark}
Adapted node edit costs can strengthen lower bounds for the GED when partial node maps are already fixed during a heuristic tree search such as \AStar.  
Typically, lower bounds for the cost of mapping the remaining nodes are computed by solving a minimum weight (perfect) matching problem
between nodes and their incident edges (``stars'') in $G$ and corresponding stars in $H$ \cite{apx-stars}.  
To tighten this bound, the partial mapping $\psi$ should be taken into account. Adapted node edit costs achieve exactly this. It suffices to substitute $\ANEC{\psi}$ for $\NEC$, while restricting edges to $\PE{V \setminus \dompsi}$ and $\PE{W \setminus \impsi}$.
\end{remark}

\subsection{Algorithm and Analysis}

\begin{algorithm}[t] 
\caption{SR-GED}\label{alg:sr-ged}
\begin{algorithmic}[1]
\Function{SR-GED}{$G$, $H$, $\NEC$, $\EEC$}
\If{$n=|V(G)| \le 1$} \Comment{in practice: larger threshold}
    \State \Return $\GED(G, H, \NEC, \EEC)$ \label{algolabel:enumeratesmall} \Comment{solve by enumeration}
\EndIf
\State Compute a small balanced separator $S \subseteq V(G)$, $S \neq \emptyset$ \label{algolabel:defs}
\State $G_r \gets G - S$
\State Let $V_1 \sqcup V_2 = V(G_r)$ be an induced balanced bipartition \label{algolabel:end-setup}
\State $u \gets \infty$ \Comment{upper bound for the $\GED$}
\ForAll{$\psi \in \inj{S}{V(H)}$} \label{algolabel:start_outer_for}
    \State $H_r \gets H - \impsi$ \label{algolabel:setup-psi}
    \State Compute the adapted node edit costs $\ANEC{\psi}$ \label{algolabel:adapt-nec}
    \State $c_\psi \gets \IEC(\psi)$ on $(G[S]$, $H[\impsi])$\Comment{induced edit costs on the separator} \label{algolabel:end-setup-psi}
    \ForAll{bipartitions $W_1 \sqcup W_2 = V(H_r)$, $|W_1| = |V_1|$} \label{algolabel:start_inner_for}
        \State $c_1 \gets \text{SR-GED}(G_r[V_1], H_r[W_1], \ANEC{\psi}, \EEC)$\label{algolabel:ged1}
        \State $c_2 \gets \text{SR-GED}(G_r[V_2], H_r[W_2], \ANEC{\psi}, \EEC)$\label{algolabel:ged2}
        \State $c_\mathrm{cut}^{+} \gets \deltaECplus(W_1)$ \Comment{cost of inserting edges on the $W_1$--$W_2$ cut in $H_r$} \label{algolabel:compute-cut}
        \State $c \gets c_\psi + c_1 + c_2 + c_\mathrm{cut}^{+}$ \label{algolabel:costsum}
        \State $u \gets \min\{u, c\}$ \label{algolabel:defu}
    \EndFor \label{algolabel:end_inner_for}
\EndFor\label{algolabel:end_outer_for}
\State \Return $u$
\EndFunction
\end{algorithmic}
\end{algorithm}

The algorithm is a recursive function that computes the GED for an instance given by graphs $G$ and $H$ and the node and edge edit cost matrices $\NEC$ and $\EEC$.
We first check whether the problem has small constant size, in which case it can be solved by enumeration in constant time (l.\ \ref{algolabel:enumeratesmall}).
Otherwise, we decompose the problem into exponentially many subproblems.
For this, we first find a separator $S$ in $G$ (l.\ \ref{algolabel:defs}), for example by exhaustive search. To guarantee termination and avoid some edge cases, we demand that $S$ be nonempty. 
The removal of $S$ yields a bipartition of nodes $V_1 \sqcup V_2 = V(G) \setminus S$.
We will later show that this bipartition can be chosen such that it is sufficiently balanced (l.\ \ref{algolabel:end-setup}).

Then, in the outer for loop (l.\ \ref{algolabel:start_outer_for}--\ref{algolabel:end_outer_for}) we try all possible ways of injectively mapping the nodes in the separator to nodes in $H$.
For each such embedding, we obtain a smaller GED problem on the remaining subgraph $G_r$ of $G$ without the separator, resp.\ $H_r$ of $H$ without the image of the separator.
To account for the edges in the cut of the separator (resp.\ image of the separator),
we adapt the node edit costs depending on the embedding (l.\ \ref{algolabel:adapt-nec}).

As $S$ is a separator, $G_r$ is the union of two disjoint subgraphs $G[V_1]$ and $G[V_2]$, which must be mapped optimally to vertex-disjoint subgraphs of $H_r$.
This is done by enumerating all bipartitions of nodes of $H_r$ in the inner for loop (l.\ \ref{algolabel:start_inner_for}--\ref{algolabel:end_inner_for}).
For a fixed bipartition $W_1\sqcup W_2=V(H_r)$, we then obtain two smaller GED problems on $(G_r[V_1], H_r[W_1])$ and $(G_r[V_2], H_r[W_2])$ which are solved by two recursive calls (l.\ \ref{algolabel:ged1}--\ref{algolabel:ged2}). We will later show that this reduction is correct.
Because there are no edges in the cut of $V_1$ and $V_2$,
we only need to sum the insertion costs of edges in the cut of $W_1$ and $W_2$ (l.\ \ref{algolabel:compute-cut}).

The induced edit costs $\IEC(\phi)$ for a mapping $\phi$ from $G$ to $H$ such that $\restr{\phi}{S} = \psi$ and $\phi[V_1] = W_1$, $\phi[V_2] = W_2$
are then given by the sum of the induced edit costs on the separator, the cost of edges in the cut,
and the minimum cost of mapping the two subgraphs of $G$ to the corresponding subgraphs of $H$ (l.\ \ref{algolabel:costsum}).
The running minimum $u$ over all of these thus gives a gradually improving upper bound,
which becomes exact once all embeddings and bipartitions have been tried.

\subparagraph*{Correctness.}
If we show that the inner for loop in lines \ref{algolabel:start_inner_for}--\ref{algolabel:end_inner_for} of Algorithm~\ref{alg:sr-ged} computes $\GED(G_r, H_r, \ANEC{\psi}, \EEC)$, the correctness of the algorithm follows from Corollary~\ref{lemma:ged-embed}.

First, we note that if $|V(G)|=|V(H)|$, then also $|V(G_r)|=|V(H_r)|$. Since $G_r$ is obtained by removing a separator from $G$, we focus on the interaction of partial node maps whose domains are not connected by any edges. 

Let $V_1\sqcup V_2 =V(G_r)$ be such that $\delta_{G_r}(V_1)=\emptyset=\delta_{G_r}(V_2)$, and let $\psi_i\in\bij{V_i}{W_i}$ for $i=1,2$ be partial node maps such that $W_1 \sqcup W_2=V(H_r)$. 
Then, $\deltaEC(\psi_1, \psi_2)$ only accounts for inserting all edges between $W_1$ and $W_2$.
In particular, it follows that $\deltaEC$ is independent of the node maps themselves.
We thus denote 
\begin{align*}
        \deltaECplus(W_i) := \sum_{w_1\in W_1}\sum_{w_2\in W_2}\EEC(\emptyset , \{ w_1, w_2 \}) , 
\end{align*} 
and get $\deltaEC(\psi_1, \psi_2)=\deltaECplus(W_i)$.

Similar to Lemma~\ref{lemma:restricted_ged_to_induced_costs}, this lets us express $\AGED{\psi_1}(G, H)$ as a GED on the remaining graph plus the cut edit costs. Since these depend only on $W_1=\impsione$, the full GED reduces to minimizing over all partitions $W_1 \sqcup (V(H_r)\setminus W_1)$ with optimum partial mappings and the corresponding cut edit costs.

\begin{lemma}[Combining GEDs of disjoint subgraphs] \label{lemma:ged-split}
    Let $H=(W,F)$ be an arbitrary graph, let $G$ be the disjoint union of two subgraphs $G_1$ and $G_2$, and let $V_i := V(G_i)$ for $i=1,2$. Then:
    \begin{equation*}
        \GED(G, H) = \mindown{W_1 \in \binom{W}{|V_1|}}{\GED(G_1, H[W_1]) + \GED(G_2, H - W_1) + \deltaECplus(W_1)}
    \end{equation*}
\end{lemma}

\begin{proof}
For a fixed $\psi_1\in\inj{V_1}{W}$, we have:
\begin{align*}
    \AGED{\psi_1}(G, H) &= \IEC(\psi_1) + \minback{\psi_2 \in \bij{V_2}{W \setminus \impsione}}{\IEC(\psi_2) + \deltaEC(\psi_1, \psi_2)}\\
    &= \IEC(\psi_1) + \deltaECplus(\impsione) + \minback{\psi_2 \in \bij{V_2}{W \setminus \impsione}}{\IEC(\psi_2)} \\
    &= \IEC(\psi_1) + \deltaECplus(\impsione) + \GED(G_2, H - \impsione)
\end{align*}

Crucially, the cut edit costs do not depend on $\psi_1$ since $G_1$ and $G_2$ are disjoint. Therefore, we can minimize over all possible images $W_1 := \impsione$ and obtain:
\begin{align*}
    \GED(G, H) 
    &= \mindown{\psi_1 \in \inj{V_1}{W}}{\AGED{\psi_1}(G, H)}\\
    &= \mindown{\psi_1 \in \inj{V_1}{W}}{\IEC(\psi_1) + \deltaECplus(\impsione) + \GED(G_2, H - \impsione)} \\
    &= \mindown{W_1 \in \binom{W}{|V_1|}} {\mindown{\psi_1 \in \bij{V_1}{W_1}}{\IEC(\psi_1) + \deltaECplus(W_1) + \GED(G_2, H - W_1)}} \\
    &= \mindown{W_1 \in \binom{W}{|V_1|}}{ \mindown{\psi_1 \in \bij{V_1}{W_1}} {\IEC(\psi_1)} + \deltaECplus(W_1) + \GED(G_2, H - W_1) } \\
    & = \mindown{W_1 \in \binom{W}{|V_1|}}{\GED(G_1, H[W_1]) + \GED(G_2, H - W_1) + \deltaECplus(W_1)} 
    \qedhere
\end{align*}
\end{proof}

This lemma shows that the inner for loop in lines \ref{algolabel:start_inner_for}--\ref{algolabel:end_inner_for} of Algorithm~\ref{alg:sr-ged} correctly computes $\GED(G_r, H_r, \ANEC{\psi}, \EEC)$. We now combine it with Corollary~\ref{lemma:ged-embed} and an induction over the recursion tree to show the correctness of the algorithm.
\newpage

\begin{theorem}[Correctness]
SR-GED computes the $\GED$.
\end{theorem}
\begin{proof}
Let $(G,H,\NEC,\EEC)$ be a (sub)problem of size $n=|V(G)|=|V(H)|$. 
If $n \le 1$, the $\GED$ is computed directly in line \ref{algolabel:enumeratesmall}.
Otherwise, we assume that for all subproblems of size less than $n$ the algorithm computes the correct value.
Then, from $S \neq \emptyset$ in line \ref{algolabel:defs} it follows that the subproblems in lines \ref{algolabel:ged1} and \ref{algolabel:ged2} are solved correctly.
As $G_r$ is the union of two disjoint subgraphs, we can apply Lemma \ref{lemma:ged-split}, and get
\begin{equation*}
    \GED(G_r,H_r,\ANEC{\psi}, \EEC) = \mindown{W_1\in\binom{V(H_r)}{|V_1|}}{c_1 + c_2 + \deltaECplus(W_1)}.
\end{equation*}

We enumerate both the separator mappings $\psi$ in the outer loop (l.\ \ref{algolabel:start_outer_for}--\ref{algolabel:end_outer_for}), as well as the images $W_1$ of the node maps on $V_1$ in the inner loop (l.\ \ref{algolabel:start_inner_for}--\ref{algolabel:end_inner_for}), and then only carry over the minimum value in line \ref{algolabel:defu}. Therefore, we have
\begin{align*}
    u &= \mindown{\psi\in\inj{S}{V(H)}}{\mindown{W_1\in\binom{V(H-\impsi)}{|V_1|}}{c_\psi + c_1 + c_2 +\deltaECplus(W_1)}}.
\end{align*}

Further, $\IEC(\psi)$ is independent of the choice of $W_1$, and applying Corollary \ref{lemma:ged-embed} then yields
\begin{equation*}
    u 
    = \mindown{\psi\in\inj{S}{V(H)}}{\IEC(\psi) + \GED(G_r, H_r, \ANEC{\psi}, \EEC)}
    = \GED(G,H) . \qedhere
\end{equation*}
\end{proof}

\subparagraph*{Runtime Analysis.}

Note that SR-GED exploits the structure of only one of the two input graphs; we take this to be $G$, leaving $H$ arbitrary. We now give a general runtime formula for SR-GED when at least one of the input graphs is \alphasbalsep{\alpha}{s}.

\begin{theorem}[Runtime] \label{thm:runtime}
Let $s: \mathbb{N} \to \mathbb{R}_{\ge 0}$ be a non-decreasing function and $\alpha \in \left[ \frac{1}{2},1 \right)$.
Let $T(n)$ denote the worst-case runtime of SR-GED over all pairs of graphs $(G, H)$ of size $|V(G)| = |V(H)| \le n$ such that $G$ is \alphasbalsep{\alpha}{s}. 
Let $n_i:=\floor{\alpha^i n}$, and let $d:=\floor{\log_{\frac{1}{\alpha}} n}$ be the depth of the recursion.
Then:
\begin{equation*}
    T(n) = \expO{2^{\frac{1}{1-\alpha}n} \cdot \prod_{i=0}^d n_i^{s(n_i)}} 
\end{equation*}
\end{theorem}

\begin{proof}
A call to SR-GED at depth $i$ enumerates $n_i \cdot (n_i-1) \cdot \ldots \cdot (n_i - s(n_i) + 1) \leq n_i^{s(n_i)}$ separator embeddings. 
Each of these results in up to $2^{n_i}$ pairs of recursive calls. The subproblem in each of these calls has size at most $\floor{\alpha n_i} \le n_{i+1}$.
We now bound the costs incurred during the reduction steps:
\begin{enumerate}
    \item Since the size bound $s(n_i)$ is known, finding a suitable separator is possible in $\bigO{n_i^{s(n_i) + 2}}$ time by enumerating all candidate sets for separators up to the given size, for each determining the set of connected components in $\bigO{n_i^2}$ time and then partitioning said connected components optimally in $\bigO{n_i^2}$ time (via the standard dynamic program for Subset Sum).
    This accounts for lines \ref{algolabel:defs}--\ref{algolabel:end-setup}.
    \item Given a fixed separator mapping $\psi$, computing $\IEC(\psi)$ and the adapted node edit costs $\ANEC{\psi}$ is possible in $\bigO{s(n_i) \cdot n_i^2} = \bigO{n_i^3}$ time, as is computing $H_r$, accounting for lines \ref{algolabel:setup-psi}--\ref{algolabel:end-setup-psi}.
    \item Given a fixed bipartition $W_1 \sqcup W_2 = V(H_r)$, restricting $G_r$ and $H_r$, as well as computing the cost of inserting edges in the cut, is possible in $\bigO{n_i^2}$. This covers the cost of lines \ref{algolabel:ged1}--\ref{algolabel:compute-cut}, except for the cost of the two recursive GED calls.
\end{enumerate}

We obtain the following recursion, where $c_1,c_2$ and $c_3$ are the constant runtime factors for the steps above, and $C',C$ are constants chosen sufficiently big:
\begin{align*}
    T(n_i) 
    &\leq c_1 \cdot n_i^2 n_i^{s(n_i)}  + n_i^{s(n_i)} \left( c_2 \cdot n_i^3 + 2^{n_i} \left( c_3 \cdot n_i^2 + 2\cdot T(\floor{\alpha n_i}) \right) \right) \\
    &\leq  n_i^{s(n_i)} \cdot \left( (c_1 + c_2) \cdot n_i^3 + 2^{n_i+1} \cdot \left( c_3 \cdot n_i^2 +  T(\floor{\alpha n_i}) \right) \right) \\
    & \leq C' \cdot n_i^{s(n_i)} \cdot 2^{n_i+1} \cdot \left( c_3 \cdot n_i^3 +  T(\floor{\alpha n_i}) \right) \\
    &\leq C \cdot n_i^{s(n_i)} \cdot 2^{n_i} \cdot T(n_{i+1})
\end{align*}

The last inequality follows since $T(n_{i+1})$ dominates $c_3 \cdot n_i^3$. Effectively, we have shown that asymptotically, the recursive GED calls dominate the remaining work.
We have $T(0) = \bigO{1}$. By induction, we then obtain
\begin{equation*}
    T(n) = \bigO{\prod_{i=0}^d\left( C \cdot 2^{n_i}\cdot  n_i^{s(n_i)}\right)}
    = \bigO{C^d \cdot \left( 2^{\sum_{i=0}^d n_i} \right)\cdot \left(\prod_{i=0}^d n_i^{s(n_i)}\right)} .
\end{equation*}

We can bound $\sum_{i=0}^d n_i $ by the geometric series $\sum_{i=0}^\infty \alpha^i n = \frac{1}{1 - \alpha} n$. Therefore, 
\begin{equation*}
    T(n) = \bigO{C^d \cdot 2^{\frac{1}{1 - \alpha}n}\cdot \prod_{i=0}^d n_i^{s(n_i)}}.
\end{equation*}

Since $d=\bigO{\log n}$, the factor $C^d$ is polynomial and we get the desired result:
\begin{equation*}
     T(n) = \expO{2^{\frac{1}{1-\alpha}n} \cdot \prod_{i=0}^d n_i^{s(n_i)}}
     \qedhere
\end{equation*}
\end{proof}

Note that the first factor of the runtime depends only on the balance factor of the separators, while the second factor depends on their size. \medskip

Next, we show that any \alphasbalsep{\beta}{s} graph $G$ is in fact \alphasbalsep{\alpha}{c\cdot s} with a balance factor $\alpha$ arbitrarily close to $\tfrac{1}{2}$, while incurring only a constant-factor increase $c$ (depending on $\beta$ and the desired balance) in the separator size.
\begin{lemma}
\label{lemma:balance-sep}
Let $\mu \in \left(0, \frac{1}{2} \right)$ be arbitrary. Let $s: \mathbb{N} \to \mathbb{R}_{\ge 0}$ be a non-decreasing function, $\beta\in(0,1)$, and let $G$ be a \alphasbalsep{\beta}{s} graph.
Denote $\alpha := \frac{1}{2} + \mu$.

Then $G$ is also \alphasbalsep{\alpha}{c \cdot s}, where $c = 2^{\ceil{\log_\beta \mu}}$.
\end{lemma}

\begin{proof}
Let $n = |V(G)|$. It suffices to find the desired separator in $G$: by the hereditary property, the same argument applies to all induced subgraphs of $G$.
Consider the following recursive scheme to produce a set of small connected components $G_1, \ldots, G_k$: Given a graph, we find a $\beta$-separator. If all connected components have size at most $\mu n$, we are done. Otherwise, we recurse on all connected components of size exceeding $\mu n$. The disjoint union of all separators found in this way then splits the graph into components of size at most $\mu n$.
The graphs $G_1, \ldots, G_k$ can then be partitioned into two sets containing at most $\alpha n$ many nodes using a simple greedy strategy.

For a very crude bound, we count how often we need to find a separator in this scheme.
Clearly, if removal of a separator causes a (sub)graph to decompose into multiple connected components, that cannot be worse compared to the case where we only get two connected components. So assume two connected components every time. The recursion tree then is a binary tree. The remaining connected components on the $i$-th level have size at most $\beta^i \cdot n$, therefore the tree has at most $l := \ceil{\log_\beta \mu}$ levels, giving a total of $c \leq 2^l$ nodes. 
\end{proof}

To achieve a truly exponential running time, the separator size bound $s(n)$ must be sufficiently small. This is summarized in the following theorem:
\begin{theorem} \label{thm:exptime}
Let $\varepsilon>0$ be arbitrary. Let $s: \mathbb{N} \to \mathbb{R}_{\ge 0}$ be a non-decreasing function, $\beta\in(0,1)$, and let $G$ be a \alphasbalsep{\beta}{s} graph. Denote $n=|V(G)|$. If $s(n) = \notbigO{\frac{n}{\log^2 n}}$, then SR-GED runs in time $\expO{(4 + \varepsilon)^n}$.
\end{theorem}

\begin{proof}
Given $\varepsilon>0$, we define $\lambda= \frac{1}{2}\log(\frac{4+\varepsilon}{4})$ and $\mu = \frac{1}{2} - \frac{1}{2 + \lambda}$. By Lemma~\ref{lemma:balance-sep} we need to increase the size $s(n)$ by a factor $c_{\mu}=2^{\ceil{\log_\beta \mu}}$ to get a balanced $\alpha := (\frac{1}{2}+\mu)$ separator.
Let $d := \floor{\log_{\frac{1}{\alpha}} n}$ and $n_i := \floor{\alpha^i n}$. 

We now consider the two factors $\expO{2^{\frac{1}{1-\alpha}n}}$ and $\expO{\prod_{i=0}^d n_i^{c_{\mu}\cdot s(n_i)}}$ given in the running time bound (Theorem \ref{thm:runtime}) individually.
For the first factor we get
\begin{align*}
    \expO{2^{\frac{1}{1-\alpha} n}} 
    &=\expO{2^{\frac{1}{\frac{1}{2}-\mu} n}}  
    =\expO{2^{(2+\lambda) n}} 
    = \expO{4^n \cdot 2^{\lambda n}} .
\end{align*}

The second factor depends on the size of the separators. Since the $n_i$ are non-increasing, we have $n_i^{c_{\mu}\cdot s(n_i)} \leq n^{c_{\mu}\cdot s(n_i)}$. 
For ease of notation, we define $t := \sum_{i=0}^d s(n_i)$, and get
\begin{align*}
     \expO{\prod_{i=0}^d n_i^{c_{\mu}\cdot s(n_i)}}= \expO{n^{\sum_{i=0}^d c_{\mu}\cdot s(n_i)}}= \expO{n^{c_{\mu}\cdot t}}= \expO{2^{c_{\mu} t \log n}} .
\end{align*}

Since $s$ is non-decreasing and $s(n) = \notbigO{\frac{n}{\log^2 n}}$, we have
\begin{equation*}
    t = \sum_{i=0}^d s(n_i)\leq (d+1) s(n) = \left( \floor{\log_{\frac{1}{\alpha}} n} + 1 \right)s(n) = \bigO{\log n \cdot s(n)} = \notbigO{\frac{n}{\log{n}}}.
\end{equation*}

As $c_\mu$ is constant for fixed $\varepsilon$ and $\beta$, we also have $c_\mu t\log n= \notbigO{n}$. 
Using the fact that $\notbigO{n}$ is dominated by $\lambda\cdot n$ for any $\lambda > 0$, we get
\begin{equation*}
    T(n)=\expO{4^n \cdot 2^{\lambda n}\cdot{2^{\notbigO{n}}}}=\expO{4^n \cdot 2^{2\lambda n}}=\expO{(4 \cdot 2^{2\lambda })^n}.
\end{equation*}

Inserting the definition of $\lambda$ then yields the desired runtime depending on $\varepsilon$.
\end{proof}

\begin{corollary}
Let $\varepsilon > 0$ be arbitrary. Let $\beta\in (0,1)$, $s(n) = \bigO{n^{1-\delta}}$ for some fixed $\delta > 0$, and let $G$ be a \alphasbalsep{\beta}{s} graph
. Then SR-GED runs in time $\expO{(4 + \varepsilon)^n}$.
\end{corollary}

We can rephrase this corollary. A graph class $\mathcal{G}$ is called \emph{hereditary} if it holds for any $G \in \mathcal{G}$ and $U \subseteq V(G)$ that $G[U] \in \mathcal{G}$.

\begin{corollary} \label{cor:runtime-hereditary}
For every $\varepsilon > 0$, SR-GED computes $\GED(G, H)$ in time $\expO{(4 + \varepsilon)^n}$ if $G$ is a member of a hereditary graph class that admits strictly sublinear balanced separators.
\end{corollary}
Unless the ETH is false, this runtime result cannot be generalized to arbitrary graphs \cite{sgi-hardness}.

\subparagraph*{Memory Usage.}

SR-GED has worst case linear recursion depth, and logarithmic recursion depth if separators are balanced. In each step, the newly created data structures (restricted graphs, updated node edit costs) have size in $\bigO{n^2}$, thus the overall auxiliary memory usage of SR-GED is bounded by $\bigO{n^3}$ resp.\ $\bigO{n^2 \log{n}}$.

\section{Applications}

Straightforward applications are to all hereditary graph classes that admit suitable separators. We briefly recall a few relevant landmark results.

\begin{theorem}[Planar Separator Theorem, Lipton and Tarjan \cite{planar-separator}] \label{thm:planar-separator}
If $G$ is a planar graph, then $G$ is \alphasbalsep{\frac{2}{3}}{2\sqrt{2n}}.
\end{theorem}

This generalizes to graphs that do not contain complete graphs of a fixed size as minors:
\begin{theorem}[Alon, Seymour, Thomas \cite{nonplanar-separator}] \label{thm:minor-separator}
If $G$ is a graph that does not contain $K_h$ as a minor, then $G$ is \alphasbalsep{\frac{2}{3}}{c_h \cdot \sqrt{n}}, where $c_h$ is a constant depending on $h$.
\end{theorem}

An alternative generalization is based on treewidth:
\begin{theorem}[Separators in Graphs of Bounded Treewidth, Robertson and Seymour \cite{treewidth-props}] \label{thm:tw-separator}
Let $\mathcal{G}$ be a hereditary graph class such that the treewidth of any graph $G\in\mathcal{G}$ is bounded by $s(|V(G)|)$ for some non-decreasing function $s: \mathbb{N} \to \mathbb{R}_{\ge 0}$. Then $\mathcal{G}$ is \alphasbalsep{\frac{2}{3}}{s+1}.
\end{theorem}

In particular, graphs $G$ with treewidth bounded by a constant $k$ are \alphasbalsep{\frac{2}{3}}{k+1}.

\begin{corollary}
Let $\varepsilon >0$ be arbitrary.
SR-GED computes $\GED(G, H)$ in $\expO{(4 + \varepsilon)^n}$ time and polynomial space if $G$ is $K_h$-minor-free for fixed $h$, or belongs to a hereditary graph class of strictly sublinearly bounded treewidth.
\end{corollary}

Computations show that the graphs in GEDLIB \cite{gedlib}, a standard collection of benchmark data sets for GED, have treewidths bounded by small constants (see Appendix \ref{appendix:tw-gedlib} for details).\medskip

As the TSP on a graph $G$ can be reduced to a GED problem between a cycle on $n$ nodes and $G$, SR-GED can be leveraged to solve TSP instances in exponential time.

\begin{theorem}
Let $\varepsilon > 0$. TSP can be solved in $\expO{(4 + \varepsilon)^n}$ time and polynomial space.
\end{theorem}

This essentially matches Gurevich and Shelah \cite{tsp-exptime-polyspace}, who showed that the TSP can be solved in time $\expO{4^n}$ and polynomial space (cf.\ \cite{tsp-exptime-polyspace-nice}). \medskip

Further, SR-GED gives rise to an exponential time algorithm for the Quadratic Assignment problem (QAP), as originally defined by Koopmans and Beckmann \cite{kb-qap}, if the interaction network of nonzero flows has a suitable structure.

\begin{theorem}[cf. Appendix \ref{appendix:qap}]
Let $\varepsilon > 0$. A QAP of size $n$ with symmetric distances can be solved in time $\expO{(4 + \varepsilon)^n}$ if the interaction network of nonzero flows is \alphasbalsep{\alpha}{\bigO{n^{1-\delta}}} for some constants $\alpha \in \left[ \frac{1}{2}, 1 \right)$, $\delta > 0$.
\end{theorem}

\section{Conclusion and Future Work}

We presented SR-GED, a novel algorithm for exact GED computation with running time $\expO{(4 + \varepsilon)^n}$ for arbitrary $\varepsilon > 0$ and polynomial space on all hereditary classes of graphs that admit strictly sublinear balanced separators, such as planar graphs. 

To improve the worst-case running time, trading time for space via dynamic programming is worth exploring, since many subproblems considered in each recursive step are similar. 
For practical applications, our theoretical algorithm requires further engineering. 
Here, limiting the number of bipartitions that need to be explored by using heuristic pruning strategies is a natural next step.
Additionally, given the recursive nature of our algorithm, it can be combined with practical state-of-the-art approaches, such as ILPs, in a straightforward way.

\bibliography{lipics-v2021-sample-article}

\newpage

\appendix

\section{Removing the Constraints on our Formulation}\label{appendix:restrictions}
We show that all the restrictions we imposed in our problem formulations can be removed with a polynomial-time preprocessing. Let $G'$ and $H'$ be two arbitrary graphs.
\begin{itemize}
    \item The given definition of GED based on node maps has been shown to be equivalent to the classic definition based on edit paths if the node and edge edit costs satisfy the triangle inequality \cite{ilp-mip-survey}. 
    If they do not, a simple cubic time preprocessing based on an All-Pairs-Shortest-Paths computation can be applied, which leaves the GED invariant but ensures that the edit costs satisfy the triangle inequality \cite{ilp-mip-survey}. 
    \item Further, the triangle inequality implies that the costs are reasonable, i.e., we can conclude that there are no node insertions or deletions in an optimum solution if $|V(G')|=|V(H')|$.    
    \item If $|V(G')| < |V(H')|$, we add $k := |V(H')| - |V(G')|$ dummy vertices $V_\varepsilon := \{\varepsilon_1, \dots, \varepsilon_k\}$ to $G'$, and define $G := G' + V_\varepsilon$. 
    We then take the edit cost $\NEC(\varepsilon_i, w)$ to be the cost of inserting $w$; this ensures that $\GED(G, H') = \GED(G', H')$. Any optimum node map from $G$ to $H'$ corresponds to an optimum (injective) node map from $G'$ to $H'$.
    Note that, given reasonable edit costs, we never want to delete nodes from the smaller graph.
    \item For non-simple graphs, loops can be absorbed into the node edit costs. Parallel edges can also be handled efficiently. Let $v_1, v_2 \in V(G')$ and $w_1, w_2 \in V(H')$. If there are any edges between $v_1$ and $v_2$, we add a single corresponding edge to $E(G)$; analogously, $E(H)$ consists of pairs of nodes that have at least one edge between them. The costs of deletions or insertions are simply the sums over the corresponding bundles of parallel edges.
    To determine the edge edit cost for a substitution $\EEC(\{v_1, v_2\}, \{w_1, w_2\})$, we solve an assignment problem to find a mapping between the bundles of parallel edges that minimizes the cost, which can be done in polynomial time.
\end{itemize}

\newpage

\section{Relationship to QAP}
\label{appendix:qap}

The (symmetric) Quadratic Assignment Problem (QAP) introduced by Koopmans and Beckmann \cite{kb-qap} -- not to be confused with a more general problem posed by Lawler \cite{lawler-qap}, which is also referred to as QAP -- can be expressed as a special case of a GED problem if fully flexible edit costs are allowed, as is the case in our definition. Because SGI can also be reduced to QAP, the lower bound of $n^{\notbigO{n}}$ under the ETH \cite{sgi-hardness} also applies to the QAP, thus again a restriction is justified if better runtimes are to be achieved. 

\begin{definition}[Quadratic Assignment Problem (QAP), Koopmans and Beckmann \cite{kb-qap}]
    Let $N := \{1, \dots, n\}$. We imagine a set of $n$ facilities is to be placed at $n$ locations.
    Let $f_{i,j}$ be the flow between facilities $i$ and $j$. Let $d_{k,l}$ be the distance between locations $k$ and $l$. Let $b_{i,k}$ denote the cost of placing the $i$-th facility at the $k$-th location.

    We wish to minimize $\sum_{i=1}^n b_{i,\varphi(i)} + \sum_{i=1}^n \sum_{j=1}^n f_{i,j} d_{\varphi(i), \varphi(j)} $ over all permutations $\varphi \in \bij{N}{N}$.
\end{definition}

\begin{theorem}
A QAP of size $n$, with a symmetric distance function, can be reduced to a GED problem of the same size in polynomial time.
\end{theorem}

\begin{proof}
We may assume $d_{k,k}=0$; nonzero $d_{k,k}$ entries can be absorbed into $b_{i,k}$.
If the distances are symmetric, the objective function simplifies to $\sum_{i=1}^n b_{i,\varphi(i)} + \sum_{i=1}^n \sum_{j=1}^n \frac{1}{2} (f_{i,j} + f_{j,i}) d_{\varphi(i), \varphi(j)}$,
which is the same as the objective function for a symmetrized flow $f'_{i,j} := \frac{1}{2}(f_{i,j} + f_{j,i})$.
Therefore, assume that $f$ is symmetric as well.    

The idea is to let $G$ be a graph representing the facilities and $H$ be a graph representing the possible locations.
We set $V = W = N$ and $E = F = \PE{V} = \PE{W}$, so $G = H$ are two complete $n$-node graphs.
The node edit costs are simply the costs of opening facilities at locations, $\NEC = b$.

The edge edit costs are defined by the cost incurred when routing the flow between facilities $i,j$ between the two locations $k,l$. We also need to account for the fact that the above formulation counts these twice. We set $\EEC(\{i,j\},\{k,l\}) := 2f_{i,j} d_{k,l}$.
\end{proof}

Note that this reduction relies crucially on the definition of the GED, where our algorithm permits us to allow arbitrary cost functions.

\begin{theorem}
A QAP of size $n$ with symmetric distances can be solved in time $\expO{(4 + \varepsilon)^n}$ if the interaction network of nonzero flows is \alphasbalsep{\alpha}{\bigO{n^{1-\delta}}} for some constants $\alpha \in [\frac{1}{2}, 1)$ and $\delta > 0$.
\end{theorem}

\begin{proof}
We need to be able to exploit a (typically sparse) structure of $f$ for our algorithm to achieve a better time complexity. Fix an edge $\{i,j\}$ with $f_{i,j} = 0$. All edit costs $\EEC(\{i,j\}, \{k,l\})$, for any edge $\{k,l\}$ in $F$, will be zero.
We remove all edges with zero flow from the facility graph by setting $E = \{\{i, j\} : f_{i,j} \neq 0\}$. Define the cost of an insertion of an edge $\{k,l\}$ to be zero (so that we do not incur any costs for distances along which there is no flow). Since $H$ is complete, no deletions are possible, hence the cost of deletions can be chosen arbitrarily under our definitions. If deletions must be given a cost, they can be forbidden by giving them prohibitively high cost (one plus the maximum over all possible substitution costs will do), reflecting the fact that all nonzero flows must be satisfied.
By \Cref{cor:runtime-hereditary}, the resulting GED instance is solvable in time $\expO{(4 + \varepsilon)^n}$.
\end{proof}

\newpage

\section{Treewidths GEDLIB}
\label{appendix:tw-gedlib}

We argue that many standard benchmark datasets do admit suitable structure for our algorithm. We compute treewidths of graphs in GEDLIB \cite{gedlib,iam,cmu} using an exact algorithm \cite{tamaki-tw}. The results are given in Table \ref{tab:tws}. We record for each dataset $\mathcal{G}$ the number of graphs $|\mathcal{G}|$, the average number of nodes ($n$) resp.\ edges ($m$), and the maximum treewidth of any graph in $\mathcal{G}$. Graphs without edges are omitted. As can be seen, all data sets consist of graphs with treewidths bounded by small constants.
\begin{table}[h!]
\begin{tabular}{lrrrrr}
\toprule
Dataset & $|\mathcal{G}|$ & max.\ $n$ & avg.\ $n$ & avg.\ $m$ & max.\ treewidth \\
\midrule
acyclic & 185 & 11 & 8.1 & 7.1 & 1 \\
AIDS & 2000 & 95 & 15.7 & 16.2 & 3 \\
AIDS-EDIT & 2214 & 33 & 32.0 & 33.4 & 2 \\
alkane & 149 & 10 & 8.9 & 7.9 & 1 \\
CMU-GED & 111 & 30 & 30.0 & 79.1 & 7 \\
Fingerprint & 3061 & 26 & 7.2 & 5.9 & 3 \\
GREC & 1100 & 24 & 11.5 & 11.9 & 3 \\
Letter & 6725 & 9 & 4.7 & 3.6 & 3 \\
mao & 68 & 27 & 18.4 & 19.6 & 2 \\
Mutagenicity & 4337 & 417 & 30.3 & 30.8 & 4 \\
pah & 94 & 28 & 20.7 & 24.4 & 4 \\
Protein & 600 & 126 & 32.6 & 62.1 & 8 \\
S-acyclic & 732 & 11 & 8.2 & 7.2 & 1 \\
S-mao & 272 & 27 & 18.4 & 19.6 & 2 \\
S-MOL & 600 & 12 & 10.0 & 9.0 & 1 \\
S-MOL-5 & 50000 & 15 & 12.5 & 11.5 & 1 \\
\bottomrule
\end{tabular}
\caption{Treewidths of datasets in GEDLIB.}
\label{tab:tws}
\end{table}

\end{document}